\documentclass[preprint,12pt]{elsarticle}

\usepackage{amssymb}
\usepackage{amsfonts,amssymb,amscd,amsmath,enumerate,verbatim}
\usepackage[latin1]{inputenc}
\usepackage{amscd}
\usepackage{latexsym}
\usepackage{mathptmx}
\usepackage{multicol}
\usepackage{setspace}
\usepackage{color,soul}
\usepackage[romanian]{babel}
\usepackage{combelow}

\newtheorem{Theorem}{Theorem}[section]
\newtheorem{Lemma}[Theorem]{Lemma}
\newtheorem{Corollary}[Theorem]{Corollary}

\newtheorem{Example}[Theorem]{Example}

\newproof{proof}{Proof}

\newcommand{\Fq}{\mathbb{F}_q}
\newcommand{\angbra}[1]{\langle{#1}\rangle}

\journal{Finite Fields and Their Applications}

\begin{document}

\begin{frontmatter}



\title{A Generalization of Quasi-twisted Codes:\\ Multi-twisted codes}

 \author[label1]{Nuh Aydin}
  \author[label2]{Ajdin Halilovi\'c}
\address[label1]{Kenyon College, Department of Mathematics, Gambier, OH 43022}
 \address[label2]{Lumina-The University of South-East Europe, \cb{S}os. Colentina 64b, 021187 Bucharest, Romania}

\begin{abstract}
Cyclic codes and their various generalizations, such as quasi-twisted (QT) codes, have a special place in algebraic coding theory. Among other things, many of the best-known or optimal codes have been obtained from these classes. In this work we introduce a new generalization of QT codes that we call multi-twisted (MT) codes and study some of their basic properties. Presenting several methods of constructing codes in this class and obtaining bounds on the minimum distances, we show that  there exist codes with good parameters in this class that cannot be obtained as QT or constacyclic codes. This suggests that considering this larger class in computer searches is promising for constructing codes with better parameters than currently best-known linear codes. Working with this new class of codes motivated us to consider a problem about binomials over finite fields and to discover a result that is interesting in its own right. 
\end{abstract}


\end{frontmatter}



\section{Introduction and Motivation}
Every linear code over a finite field $\mathbb{F}_q$ has three basic parameters: the length ($n$), the dimension ($k$), and the minimum distance ($d$) that determine the quality of the code. One of the most important and challenging problems of coding theory is a discrete optimization problem:  determine the optimal values of these parameters and construct codes whose parameters attain the optimal values. This optimization problem is very difficult. In general, it is only solved for the cases where either $k$ or $n-k$ is small. There is a database of best known linear codes with upper bounds on minimum distances that is available online \cite{table}. The database is updated as new codes are discovered and reported by researchers. 
   
Computers are often used in searching for codes with best parameters but there is an inherent difficulty: computing the minimum distance of a linear code is computationally intractable (NP-hard) \cite{NPhard}. Since it is not possible to conduct exhaustive searches for linear codes if the dimension is large, researchers often focus on promising subclasses of linear codes with rich mathematical structures. A promising and fruitful approach has been to focus on the class of quasi-twisted (QT) codes which includes cyclic, constacyclic, and quasi-cyclic (QC) codes as special cases.  This class of codes is known to contain many codes with good parameters. In the last few decades, a large number of record-breaking QC and QT codes have been constructed (e.g. \cite{qc1}-\cite{qc7}). The search algorithm  introduced in \cite{qtmain} has been highly effective and used in several subsequent works (\cite{qc3}-\cite{qc7}).

In this work, we introduce a new generalization of QT codes that we call multi-twisted (MT) codes. It turns out that this class  also generalizes more recently introduced classes of double cyclic codes (\cite{qc1,qc2}), QCT codes (\cite{qct}), and GQC codes (\cite{gqc}). After deriving some of their algebraic properties and obtaining a lower bound on the minimum distance, we show that from this class we can obtain linear codes with best-known or optimal parameters that cannot be obtained from the smaller classes of constacyclic or QT codes. 

Before introducing this new class of codes, we recall some fundamental  results about constacyclic and QT codes that will be needed later.

\section{Constacyclic and Quasi-twisted Codes}

Constacyclic codes are very well-known in algebraic coding theory. Let $a\in \Fq^{*}=\Fq\setminus \{0\}$. A
linear code $C$ over a finite field $\Fq$ is called constacyclic with shift constant $a$ if it is closed under
the constacyclic shift, i.e. for any $(c_0, c_1, \ldots, c_{n-1})\in C$, $T_a(c_0, c_1, \ldots, c_{n-1}) :=
(c_{n-1}, c_0, c_1, \ldots, c_{n-2})\in C$.  When $a=1$, we obtain the very important special case of  cyclic codes. Many well-known  codes are instances of cyclic codes.

Under the usual  isomorphism $\pi : \Fq^n \to \Fq[x] / \langle x^n - a\rangle$, where $\pi(u) = u_0 + u_1 x + u_2 x^2 + \ldots + u_{n-1}
x^{n-1}$, for $u=(u_0,u_1,\dots,u_{n-1}) \in \Fq^n$, it is
well-known that a constacyclic code is an ideal in the ring $\Fq[x]/ \langle x^n-a \rangle$. Moreover, for every constacyclic code $C$ there is a unique, monic polynomial of least degree in $C$ that generates $C$, i.e. 
$C=\langle g(x)\rangle=\{ f(x)g(x)\mod x^n-a: f(x)\in \Fq[x] \}$. This standard generator is a divisor of $x^n-a$, so that $x^n-a=g(x)h(x)$, for some $h(x)\in \Fq[x]$ which is called the check polynomial of $C$. 
Note that the set of all codewords can be described as
 $C=\{ f(x)g(x): f(x)\in \mathbb{F}_q[x] \text{ and } \deg(f(x))<\deg(h(x)) \}$, i.e. the set $\{ g(x),xg(x),x^2g(x),\dots,x^{k-1}g(x) \}$ is a basis for $C$ where $k=\deg(h(x))$.  

A constacyclic code has many other generators as well and we have a complete characterization of them  as follows.

\begin{Lemma} \cite{qtmain}  Let $C=\langle g(x)\rangle$ be a constacyclic code of length $n$ and shift constant $a$ with canonical generator $g(x)$ and check polynomial $h(x)$ so that $x^n-a=g(x)h(x)$. Then $C=\langle g^{\prime}(x)\rangle$ if and only if $g^{\prime}(x)$ is of the form $g(x)p(x)$ with $\gcd(p(x),h(x))=1$.
\end{Lemma}

A linear code $C$ is said to be $\ell$-quasi-twisted ($\ell$-QT) if, for a
positive integer $\ell$, it is invariant under $T_a ^{\ell}$, that is, whenever $(c_0,c_1,\dots,c_{n-1})\in C$,
then $(ac_{n-\ell},\dots,ac_{n-1},c_0,c_1,\dots,c_{n-\ell-1})\in
C$ as well. It is important to note that when $\gcd(\ell,n)=1$ we obtain
constacyclic codes. We  therefore assume that $\ell\mid n$, the
case $\ell=1$ corresponding to constacyclic codes.

It is well known that algebraically a QT code $C$ of length $n=m\ell$ with shift constant $a$ is an $R$-submodule of $R^\ell$ where $R=\Fq[x]/ \langle x^m-a \rangle$. If $C$ has a single generator of the form $(f_1(x),f_2(x),\dots,f_{\ell}(x))$ then it is called a 1-generator QT code, otherwise a multi-generator QT code. Most of the literature on QT codes focuses on the 1-generator case. We will do the same in this work.

A certain type of 1-generator QT codes, sometimes called degenerate QT codes, is particularly useful and promising when searching for new linear codes.  This is due to the lower bound on the
minimum distance given in the following theorem. 
\begin{Theorem}[see \cite{qtmain}]\label{thm:lowerboundQT}
\label{main} Let $C$ be a $1$-generator $\ell$-QT code of
length $n=m\ell$ with a generator of the form:
\begin{equation}\label{eq:one-generator-special-form}
(f_1(x)g(x),f_2(x)g(x),\cdots,f_{\ell}(x)g(x)),
\end{equation}
where $g(x),f_i(x) \in {\Fq[x]}/{\angbra{x^m-a}}$, such
that $x^m-a=g(x)h(x)$ and $f_i(x)$ is relatively prime
to $h(x)$  for each $i=1,\dots,\ell$. Then $C$ is an $[n,k,d']_q$-code where
$k=m-\deg(g(x))$, $d'\geq \ell\cdot d(C_g)$, and $d(C_g)$
denotes the minimum distance of the constacyclic code of length $m$
generated by $g(x)$.
\end{Theorem}

 A proof of this theorem is given in \cite{qtmain}. In reality, the actual
 minimum distance of $C$ is often considerably larger than the lower
 bound given by the theorem. Researchers designed algorithms and conducted computer searches based on this theorem and discovered  many new linear codes (\cite{qc3}-\cite{qc7}).

\section{Multi-twisted Codes}

We propose an even more general class of linear codes, which we call
\textit{multi-twisted} (MT) codes.  A \textit{multi-twisted
module} $V$ is an $\Fq[x]$-module of the form 
\[V = {\displaystyle\prod_{i=1}^\ell \Fq[x] / \langle x^{m_i} - a_i \rangle},\] where $a_i
\in \Fq\setminus \{0\}$ and $m_i $ are (possibly distinct) positive integers. An MT code is an $\Fq[x]$-submodule of a multi-twisted module $V$. Equivalently, we can define an MT code in terms of the shift of a codeword. Namely, a linear code $C$ is multi-twisted if for any codeword 
$$\vec{c}=(c_{1,0},\dots, c_{1,{m_1-1}}; c_{2,0},\dots, c_{2,{m_2-1}};\dots;c_{\ell,0},\dots, c_{\ell,{m_{\ell}-1}})\in C,$$
\noindent its multi-twisted shift
$$(a_1c_{1,m_1-1},c_{1,0},\dots, c_{1,{m_1-2}}; a_2c_{2,m_2-1},c_{2,0},\dots, c_{2,{m_2-2}};\dots; a_{\ell}c_{\ell,m_{\ell}-1},\dots, c_{\ell,{m_{\ell}-2}})$$

\noindent is also a codeword. If we identify a vector $\vec{c}$ with $C(x)=(c_1(x),c_2(x),\dots,c_{\ell}(x))$ where $c_i(x)=c_{i,0}+c_{i,1}x+\cdots+c_{i,m_i-1}x^{m_i-1}$, then the MT shift corresponds to $xC(x)=(xc_1(x) \mod x^{m_1}-a_1, \dots, xc_{\ell}(x) \mod x^{m_{\ell}}-a_{\ell})$.

Note that cyclic, constacyclic, QC, and QT  codes are all (permutation equivalent to) special cases of MT codes. For example, QT codes are obtained as a special case when 
$m_1=m_2=\dots=m_{\ell}$ and $a_1=a_2=\dots=a_{\ell}$. Moreover, more recently introduced classes of codes called generalized quasi-cyclic (GQC) codes \cite{gqc}, QCT codes \cite{qct}, and double cyclic codes \cite{dc1,dc2} can be viewed as special cases of MT codes. 

 An MT code is a one-generator code if it is generated by a single element of $V$. This work will focus primarily on one-generator MT codes.

Our next goal about MT codes is to find a lower bound on the minimum distance of a 1-generator MT code similar to the one in Theorem \ref{thm:lowerboundQT}. This leads to considering the greatest common divisor of two binomials $x^{n_1}-a_1$ and $x^{n_2}-a_2$. Working on this question, we discovered a result about the greatest common divisor of two binomials $x^{n_1}-a_1$ and $x^{n_2}-a_2$ over $\Fq$ which we believe is a new result about polynomials over finite fields, and interesting in its own right. We state and prove this result in the next section.

\section{A Result About Binomials over Finite Fields}

Considering generators of MT codes brings up the problem of determining the greatest common divisor of two binomials of the form $x^{n_1}-a_1$ and $x^{n_2}-a_2$. We found that the $\gcd$ of two such polynomials is either 1, or another binomial of the same form. The precise statement and a proof are as follows. 

\begin{Theorem}\label{gcd}
Let $P=\{f\in \mathbb{F}_q[x]: f=x^m-a \text{~with~} a\in \mathbb{F}_q^{\ast}, m\in\mathbb{N}$, $\gcd(m,q)=1\}.$ Then for $f,g\in P$, {\rm gcd}$(f,g)$ is either 1 or of the form $x^{{\rm gcd}({\rm deg}(f),{\rm deg}(g))}-a$, for some $a\in \mathbb{F}_q$. In particular, {\rm gcd}$(f,g)$ $\in P\cup\{1\}$.
\end{Theorem}

\begin{proof}
We will let  $ord_m(q)$ denote the multiplicative order of $q \mod m$, i.e. it is the smallest positive integer $k$ such that $q^k \equiv 1 \mod m$. For a non-zero element $a\in \mathbb{F}^{*}_q$, $|a|$ denotes the order of $a$ in the multiplicative group of $(\mathbb{F}^{*}_q,\cdot)$. 
Let $f=x^{n_1}-a_1$ and $g=x^{n_2}-a_2$ for some $n_1, n_2 \in \mathbb{N}$ and non-zero elements $a_1,a_2\in \mathbb{F}_q$ with $r_1=|a_1|$ and $r_2=|a_2|$. Let $s_1=ord_{n_1r_1}(q)$ and $s_2=ord_{n_2r_2}(q)$. It is known that $f$ and $g$ split into linear factors in the extensions $F_{q^{s_1}}$ and $F_{q^{s_2}}$, respectively. In order to find a common extension of these two fields, consider $s=ord_{n_1r_1n_2r_2}(q)$. We claim that $F_{q^{s_1}}, F_{q^{s_2}}\subseteq F_{q^s}$, that is, $s_1|s$ and $s_2|s$. Indeed, since $n_1r_1n_2r_2|q^s-1$, we have, in particular, $n_1r_1|q^s-1$. Since $s_1=ord_{n_1r_1}(q)$, it follows that $s_1|s$. Similarly we obtain $s_2|s$.

Let $\zeta\in F_{q^s}$ be a primitive $n_1n_2$th root of unity. Then $\zeta^{n_2}$ and $\zeta^{n_1}$ are primitive $n_1$th and $n_2$th roots of unity, respectively. 

If $f$ and $g$ do not have a common root, then gcd$(f,g)=1$, and we are done. Suppose now that there exists a common root $\delta$ of $f$ and $g$, that is,
$\delta$ is an $n_1$th root of $a_1$ and $n_2$th root of $a_2$. 
It follows that the roots of $f$ are \[\delta, \delta\zeta^{n_2}, \delta(\zeta^{n_2})^2,\dots,\delta(\zeta^{n_2})^{n_1-2} ~{\rm and}~ \delta(\zeta^{n_2})^{n_1-1},\] and the roots of $g$ are \[\delta, \delta\zeta^{n_1}, \delta(\zeta^{n_1})^2,\dots,\delta(\zeta^{n_1})^{n_2-2} ~{\rm and}~ \delta(\zeta^{n_1})^{n_2-1}.\]

The set of roots of $\gcd(f,g)$ is the intersection of the two sets above. Note that these sets are actually cosets of the multiplicative subgroups\\
 $\{1, \zeta^{n_2}, (\zeta^{n_2})^2,\dots,(\zeta^{n_2})^{n_1-1}\}$ and $\{1, \zeta^{n_1}, (\zeta^{n_1})^2,\dots,(\zeta^{n_1})^{n_2-1} \}$ of $F_{q^s}$, respectively. It follows that the set of roots of $\gcd(f,g)$ is a coset of  the intersection of these subgroups. Therefore, $\zeta^{n_1n_2/d }$ is a generator (primitive element) of the subgroup of the intersection where $d=\gcd(n_1,n_2)$. Hence, the roots of $\gcd(f,g)$ are \[\delta, \delta\zeta^{n_1n_2/d}, \delta(\zeta^{n_1n_2/d})^2,\dots,\delta(\zeta^{n_1n_2/d})^{{d-1}}.\]
Therefore, $\deg(\gcd(f,g)) = \gcd(n_1,n_2) = \gcd(\deg(f), \deg(g))$. Finally, we show that $a:=\delta^d\in \mathbb{F}_q$ to prove that $\gcd(f,g)$ is of the desired form $x^m-a$, hence completing the proof. Write $n_1=dt_1$ and $n_2=dt_2$ for some relatively prime integers $t_1,t_2$. 
Since $\gcd(t_1,t_2)=1$, there exist integers $u,v$ such that $ut_1+vt_2=1$. We know that $\delta^{dt_1}=a_1$ and $\delta^{dt_2}=a_2$, which implies
 \[ a=\delta^d=\delta^{udt_1+vdt_2}=a_1^u\cdot a_2^v \in \mathbb{F}_q. \]

\end{proof}

\begin{Example}
 Over $\mathbb{F}_7$, $\gcd(x^{10} -4, x^{15} -1)=x^5-2$,   $\gcd(x^{11}-5,x^{16}-4)= x-3$ and $\gcd(x^{12}-3,x^{15}-4)=1$.
\end{Example}

This theorem has some implications for constacyclic codes. Suppose $a_1,a_2\in \Fq$ and $n_1,n_2\in \mathbb{Z}^{+}$ are such that 
$\gcd(x^{n_1}-a_1,x^{n_2}-a_2)=x^m-a$ for some non-zero $m \in \mathbb{Z}$ and $a\in \Fq$. Any constacyclic code $C$ of length $m$ with shift constant $a$ has a generator $g(x)$ that divides $x^m-a$, hence $g(x)|(x^{n_1}-a_1)$ and $g(x)|(x^{n_2}-a_2)$. Therefore, we observe that the polynomial $g(x)$ can also be regarded as the (standard) generator of a constacyclic code $C_1$ of length $n_1$ with shift constant $a_1$  as well as the generator  of a constacyclic code $C_2$ of length $n_2$ with shift constant $a_2$.

\section{More on Constructions of MT Codes and Their Parameters }

We frequently have $\gcd(x^{n_1}-a_1,x^{n_2}-a_2)=1$. Consider a 1-generator MT code $C$ in this case with a generator of the form $\langle g_1(x), g_2(x)\rangle$ where $g_1(x)|x^{n_1}-a_1$ and $g_2(x)|x^{n_2}-a_2$. Clearly, $\gcd(g_1(x),g_2(x))=1$ as well. If $C_1=\langle g_1(x) \rangle$ with parameters $[n_1,k_1,d_1]$ and $C_2=\langle g_2(x) \rangle$ with parameters $[n_2,k_2,d_2]$, then 
we observe that $C$ has parameters $[n_1+n_2,k_1+k_2,\min\{d_1,d_2\}]$. We formally prove this in the next theorem. 
\begin{Theorem}\label{prop}
 Let $n_1,n_2\in \mathbb{Z}^{+}, a_1,a_2\in \Fq^*$ be such that $\gcd(x^{n_1}-a_1,x^{n_2}-a_2)=1$. 
Let $x^{n_1}-a_1=g_1(x)h_1(x)$ and $x^{n_2}-a_2=g_2(x)h_2(x)$ with $k_1=\deg(h_1(x)), k_2=\deg(h_2(x))$. Let $C_1=\langle g_1(x)\rangle$ and $C_2=\langle g_2(x)\rangle$ with parameters $[n_1,k_1,d_1]$ and $[n_2,k_2,d_2]$, respectively. Then the MT code $C$ with generator $\langle g_1(x), g_2(x)\rangle$ has parameters $[n,k,d]$, where $n=n_1+n_2$, $k=k_1+k_2$, and $d=\min\{d_1,d_2\}$.
\end{Theorem}

\begin{proof}

The assertions on length of the MT code $C$ is clear. To prove the assertion on dimension, we show that the set $S=\{x^i\cdot (g_1(x),g_2(x)): 0\leq i\leq k-1\}$, where $k=k_1+k_2$, is a basis for $C$. Suppose that $f(x)\cdot (g_1(x),g_2(x))=(f(x)g_1(x)\mod x^{n_1}-a_1$, $f(x)g_2(x)\mod x^{n_2}-a_2)=0.$ Then
 $f(x)g_1(x) \equiv 0 \mod x^{n_1}-a_1 $ and $f(x)g_2(x) \equiv 0 \mod x^{n_2}-a_2 $. Therefore, $h_1(x)|f(x)$ and $h_2(x)|f(x)$. Since $h_1(x)$ and $h_2(x)$ are relatively prime (because $\gcd(x^{n_1}-a_1,x^{n_2}-a_2)=1$), we have $h_1(x)h_2(x)|f(x)$, and hence $\deg(f(x))\geq k$. This implies the linear independence of vectors in $S$. 
 
 It remains to show that $S$ is a set of generators for $C$. For this it suffices to show that for every $f(x)\in \Fq[x]$ with $\deg(f(x))\geq k$ there exists $r(x)\in \Fq[x]$ with $\deg(r(x))< k$ such that $f(x)\cdot (g_1(x),g_2(x))=r(x)\cdot (g_1(x),g_2(x)).$ 
So let $f(x)\in \Fq[x]$ be an arbitrary polynomial with $\deg(f(x))\geq k$. We can write $f(x)$ as $f(x)=h_1(x)h_2(x)q(x)+r(x)$, for some $q(x),r(x)\in \Fq[x]$ where $\deg(r(x))< \deg(h_1(x)h_2(x))=k$. It follows that
\begin{align*}
f(x)\cdot (g_1(x),g_2(x))&=(f(x)g_1(x)\mod x^{n_1}-a_1, f(x)g_2(x)\mod x^{n_2}-a_2)\\
                          &=(r(x)g_1(x)\mod x^{n_1}-a_1, r(x)g_2(x)\mod x^{n_2}-a_2)\\
                          &=r(x)\cdot (g_1(x),g_2(x)).
\end{align*}
Finally, to show that $d=\min\{d_1,d_2\}$, let $u(x)=t(x)g_1(x)\in C_1$ be a codeword of minimum weight in $C_1$. Since $\gcd(h_1(x),h_2(x))=1$, we have $C_1=\langle g_1(x)\rangle = \langle h_2(x)g_1(x)\rangle$. Hence,  $u(x)=t^{\prime}(x)h_2(x)g_1(x)$ for some $t^{\prime}(x) \in \Fq[x]$ (degree of which can be taken to be $<k_1$). Letting $f(x)=t^{\prime}(x)h_2(x)$ and considering the codeword $f(x)\cdot (g_1(x),g_2(x))$, we find a codeword $(u(x),0)\in C$  of weight $d_1$ in $C$. We can similarly show that there exists a codeword of weight $d_2$ in $C$ as well.    
\end{proof}

 Given this result, we cannot hope to find codes with good parameters under the conditions of the above proposition. We need to look for alternative ways to construct MT codes with potentially high minimum distances. Before coming to more promising constructions, let us first disqualify another case where no MT codes with high minimum distances can be expected.

\begin{Theorem}\label{thm1}
Let $a_1,a_2\in \Fq^*$, $n_1\not = n_2\in \mathbb{Z}^{+}$ be such that 
$\gcd(x^{n_1}-a_1,x^{n_2}-a_2)=x^m-a=g(x)h(x)$,  for some $m \in \mathbb{Z}^{+}, a\in \Fq$ and let $x^{n_1}-a_1=g(x)h(x)h_1(x)$ and $x^{n_2}-a_2=g(x)h(x)h_2(x)$ with $k=\deg(h(x)), k_1=\deg(h_1(x))$, $k_2=\deg(h_2(x))$. Let $C$ be the MT code of length $n_1+n_2$ with generator $\langle g(x)p_1(x), g(x)p_2(x) \rangle$, where $\gcd(h(x)h_1(x),p_1(x))=1$ and $\gcd(h(x)h_2(x),p_2(x))=1$. Then $C$ has parameters $[n_1+n_2,k+k_1+k_2,d]$ with $d\leq 2$.
\end{Theorem}

\begin{proof}
An argument identical to the one used in Proposition \ref{prop} shows that $S=\{x^i\cdot (g(x)p_1(x), g(x)p_2(x)): 0\leq i\leq k+k_1+k_2-1\}$ is a basis for $C$, and hence proves the assertion about the dimension.

In order to see that minimum distance is at most 2,  we first assume, w.l.o.g., that $n_1 > n_2$ and observe that $C$ is also generated by $\langle g(x), g(x)p_2(x)\rangle$. Then we  note that $h(x)h_2(x)\cdot(g(x), g(x)p_2(x))=(x^{n_2}-a_2, 0)$, which is a codeword of weight $2$.
\end{proof}

Note that the theorem above points out a significant difference between QT and MT codes, since many good codes from the class of QT codes have generators of the form $\langle g(x), g(x)p(x) \rangle$, where $\gcd(h(x),p(x))=1$ \cite{qtmain}. Hence, we need to consider MT codes with generators of different form to look for codes with potentially high minimum distances. The next theorem presents such a subclass of MT codes. Afterwards we will introduce a class of promising subcodes of MT codes. 

\begin{Theorem}\label{mainthm}
Let $a_1,a_2\in \Fq^*$, $n_1,n_2\in \mathbb{Z}^{+}$ be such that $\gcd(x^{n_1}-a_1,x^{n_2}-a_2)=x^m-a=g(x)h(x)$,  for some $m \in \mathbb{Z}^{+}, a\in \Fq$ and let $x^{n_1}-a_1=g(x)h(x)h_1(x)$ and $x^{n_2}-a_2=g(x)h(x)h_2(x)$ with $k=\deg(h(x)), k_1=\deg(h_1(x))$, $k_2=\deg(h_2(x))$.
Let $C_1=\langle h(x)h_1(x)\rangle$ be a constacyclic code with shift constant $a_1$ and parameters $[n_1,m-k,d_1]$ and let $C_2=\langle h(x)h_2(x)\rangle$ be a constacyclic code with shift constant $a_2$ and parameters $[n_2,m-k,d_2]$. Then an MT code $C$ with a generator of the form $\langle p_1(x)h(x)h_1(x),p_2(x)h(x)h_2(x)\rangle$, where $\gcd(p_i(x),g(x))=1$ for $i=1,2$, has parameters $[n_1+n_2,m-k,d]$, where $d\geq d_1+d_2$.
\end{Theorem}

\begin{proof}
The assertion on dimension can be proven as in the previous theorems. To see why $d\geq d_1+d_2$, let $c:=f(x)\cdot(p_1(x)h(x)h_1(x),p_2(x)h(x)h_2(x))$, for some $f(x)\in\Fq[x]$, be an arbitrary codeword in $C$. We observe that the first component of $c$ is 0 if and only if $g(x)|f(x)$ if and only if the second component of $c$ is 0. Therefore, in $C$ there are no codewords with only 1 non-zero component, and hence $d\geq d_1+d_2$.
\end{proof}

We now give an example of an MT code with parameters of a best-known code, obtained using the theorem above. It shows that the actual minimum distance of an MT code can be significantly larger than the theoretical lower bound. 

\begin{Example}
Let $q=3, n_1=20, n_2=40, a_1=2$ and $a_2=1$ in Theorem \ref{mainthm}. We have 
$\gcd(x^{n_1}-a_1,x^{n_2}-a_2 )= x^{20}-2=g(x)h(x)$ where $g(x)=x^6 + x^5 + x^4 + 2x + 2$. The constacyclic codes $C_1$ and $C_2$ generated by $\langle h(x)h_1(x) \rangle$ and $\langle h(x)h_2(x) \rangle$ as described in Theorem \ref{mainthm}  have parameters $[20, 6, 9]$ and $[40, 6, 18]$ respectively. Then the MT code $C$ with a generator in the form given by the theorem has parameters $[60,6,d]$ with $d\geq 27$. We found that for  $p_1(x)=x^3 + x^2 + 2x$ and
$p_2(x)=x^5 + x^4 + x^3 + x^2 + 2x + 1$, the minimum weight is actually 36. Hence we obtain a ternary $[60,6,36]$ code which is a best known code for its parameters \cite{table}.  
\end{Example}

Another way to obtain codes with larger minimum distances is to consider subcodes of MT codes.

\begin{Theorem}\label{mainthm2}
Let $C_1=\langle g_1(x)\rangle$ be a constacyclic code with shift constant $a_1$ and parameters $[n_1,k,d_1]$ and let $C_2=\langle g_2(x)\rangle$ be a constacyclic code with shift constant $a_2$ and parameters $[n_2,k,d_2]$. Let $x^{n_1}-a_1=g_1(x)h_1(x)$ and $x^{n_2}-a_2=g_2(x)h_2(x)$ with $\deg(h_1(x))=\deg(h_2(x))=k$. Then an MT code $C$ with a generator of the form $\langle p_1(x)g_1(x),p_2(x)g_2(x)\rangle$, where $\gcd(p_i(x),h_i(x))=1$ for $i=1,2$, has a subcode $C'$ with parameters $[n,k,d]$, where $n=n_1+n_2$ and $d\geq d_1+d_2$. 
\end{Theorem}

\begin{proof}
   We consider the subcode $C'$ of $C$ of dimension $k$, generated by 
   \begin{center}
   $\{x^i\cdot (p_1(x)g_1(x), p_2(x)g_2(x)): 0\leq i\leq k-1\}$,
   \end{center}
 where $x^i\cdot (p_1(x)g_1(x), p_2(x)g_2(x))=(x^i p_1(x)g_1(x), x^i p_2(x)g_2(x))$. We have that every codeword of $C'$ is of the form \[a(x)\cdot (p_1(x)g_1(x),p_2(x)g_2(x))=(a(x)p_1(x)g_1(x),a(x)p_2(x)g_2(x)),\] where $a(x)$ is a polynomial of degree $<k$. To prove $d\geq d_1+d_2$, it suffices to show that for a non-zero polynomial $a(x)$ both components of the codeword $(a(x)p_1(x)g_1(x),a(x)p_2(x)g_2(x))$ are non-zero. But this is true since $\deg(a(x))<k=\deg(h_1)$ and $\gcd(p_1(x),h_1(x))=1$, and therefore $x^{n_1}-a_1$ does not divide $a(x)p_1(x)g_1(x)$. Similarly, $x^{n_2}-a_2$ does not divide $a(x)p_2(x)g_2(x)$. 
\end{proof}

Note that a straightforward generalization of the argument above proves the following

\begin{Theorem}
Let $C_i=\langle g_i(x)\rangle$ be a constacyclic code with shift constant $a_i$ and parameters $[n_i,k,d_i]$ for $i=1,\dots, \ell$. Let $x^{n_i}-a_i=g_i(x)h_i(x)$ with $\deg(h_i(x))=k$. Then an MT code $C$ with a generator of the form 
\[\langle p_1(x)g_1(x),p_2(x)g_2(x),\dots,p_\ell(x)g_\ell(x)\rangle,\] 
where $\gcd(p_i(x),h_i(x))=1$ for $i=1,\dots, \ell$, has a subcode with parameters $[n,k,d]$, where 
$\displaystyle{n=\sum_{i=1}^{i=\ell} n_i}$ and $\displaystyle{d\geq \sum_{i=1}^{i=\ell}d_i}$. 
\end{Theorem}

Finally, we highlight a special case of Theorem \ref{mainthm2} which may be useful. 

\begin{Corollary}\label{gcd1}
Let $n_1<n_2$,  and $x^{n_2}-a_2=g(x)h(x)$, where $\deg(h(x))=n_1$. Let $p_1(x)$ and $p_2(x)$ be any two polynomials such that $\gcd(p_1(x),x^{n_1}-a_1)=1$ and  $\gcd(p_2(x),h(x))=1$. Let $C$ be the MT subcode generated by $\{x^i\cdot (p_1(x), g(x)p_2(x)): 0\leq i < n_1\}$. Then $C$ has parameters $[n,k,d]$, where $n=n_1+n_2, k=n_1$, and $d\geq d_2+1$, where $d_2$ is the minimum weight of the constacyclic code of length $n_2$, shift constant $a_2$, and generated by  $g(x)p_2(x)$.
\end{Corollary}

\section{Examples}

Finally, we present a few examples of subcodes of MT codes with good parameters. All of these codes have parameters of a best-known code, and some of them are optimal. Moreover, some of them cannot be obtained as a constacyclic or QT code. Considering the class of MT codes and their subcodes in a computer search to discover new linear codes is both more promising and more challenging than in the QT case because the search space is larger. 

\begin{Example}
Let $q=7$, $n_1=7$, $n_2=16$, and $a_1=a_2=1$. We have $\gcd(x^7-1,x^{16}-1)=1$.
By letting $g(x)= x^9 + 6x^8 + x^7 + 2x^6 + x^5 + 5x^4 + 3x^3 + x^2 + 2x + 6$, 
$p_1(x) =x^6 + 2x^5 + x^4 + 2x^3 + 4x^2$, and $p_2(x) =2x^6 + 6x^4 + x^3 + x^2 + 6x + 1$ in Corollary \ref{gcd1}, we obtain a code with parameters $[23,7,13]_7$. According to the database \cite{table}, this is the parameters of a best-known code. 
\end{Example}

\begin{Example}
Let $q=7$, $n_1=4$, $n_2=50$, $a_1=2$ and $a_2=3$. We have $d(x)=\gcd(x^4-2,x^{50}-3)=x^2-3$. In this case, the gcd is greater than $1$ but we can still take $g_1(x)=1$ as the generator of a constacyclic code so that $C_1=\langle 1 \rangle$ is the trivial code $[4,4,1]$. Let  $g_2(x)= x^{46} + x^{45} + 5x^{44} + 5x^{43} + 5x^{42} + 3x^{41} + 3x^{39} + 2x^{38} + x^{37} + 4x^{36} + x^{35} + 2x^{34} + 2x^{33} + 5x^{32} + 3x^{31} + 4x^{30} + 6x^{29} + 2x^{27} + 5x^{26} + x^{25} + 6x^{24} + 3x^{22} + 5x^{21} + 5x^{20} + 6x^{19} + x^{17} + 2x^{16} + x^{15} + 5x^{14} + 6x^{13} + 4x^{12} + 6x^{11} + 2x^{10} + 5x^9 + 2x^8 + 2x^7 + 
4x^5 + 6x^4 + 4x^3 + 5x^2 + 3x + 2$. Then $g_2(x)|(x^{50}-3)$ and it generates a constacyclic code $C_2$ with parameters $[50,4,42]$.  The MT subcode generated by $\{x^i\cdot (g_1(x), p_2(x)g_2(x)): 0\leq i\leq 3\}$, where $p_2(x) =x^3 + 2x^2 + x + 5$, which is relatively prime with the check polynomial of $C_2$, has parameters $[54,4,44]_7$. This turns out to be an \textit{optimal} code \cite{table}.
\end{Example}

\begin{Example}
Let $q=5$, $n_1=19$, $n_2=34$, $a_1=2$ and $a_2=2$. The cyclic code $C_1$ of length 19 generated by 
$g_1(x)=x+4$ has parameters $[19, 18, 2]$. The constacyclic code of length 34 and shift constant 2 generated by $g_2(x)= x^{16} + 2x^{15} + 3x^{14} + 2x^{13} + 4x^{12} + 3x^{11} + 2x^9 + 4x^7 + 4x^5 + 4x^4 + 4x^3 + 2x^2 + x + 1$ has parameters $[34, 18, 10]$. By Theorem \ref{mainthm2} we know that an MT subcode generated by $\{x^i\cdot (g_1(x),p_2(x)g_2(x)): 0\leq i\leq 17\}$, where $\gcd(p_2(x),h_2(x))=1$ will have paramaters $[53,18,d]$ with $d\geq 12$.  We have found that for $p_2(x) =3x^{17} + 2x^{16} + 3x^{15} + 2x^{14} + 3x^{13} + 2x^{12} + 4x^{11} + 2x^{10} + 4x^9 + 2x^8 + 3x^7 + 3x^6 + 4x^5 + 4x^4 + 4x^3 + 4x^2 + 1$ the resulting code $C$ has parameters $[53, 18, 21]$, which is the parameters of a best known code \cite{table}. We would like to point out that since $n=53$ is a prime number, it is not possible to obtain a code of length 53 and dimension 18 from the class of QT codes with index greater than 1. Moreover, from the factorization of the polynomial $x^{53}-a$, for any $a\in \mathbb{F}^*_5$, we see that it is not possible to obtain a code with these parameters from the class of constacyclic codes either.  
\end{Example}

\begin{Example}
Let $q=3$, $n_1=13$, $n_2=20$, $a_1=1,$ and $a_2=2$. The 1-generator ternary MT subcode $C$ of dimension 12 generated by
$\{x^i\cdot (g_1(x),p_2(x)g_2(x)): 0\leq i\leq 11\}$, where $g_1(x)=x+2$, $g_2(x)=x^8 + x^7 + 2x^6 + 2x^5 + x^4 + 2x^3 + 2x^2 + x + 1$, and $p_2(x)=2x^{11} + x^{10} + 2x^9 + x^8 + 2x^7 + x^6 + x^5 + 2x^4 + 2x^3 + x^2 + 2x$, has parameters $[33,12,12]$, which means that $C$ is a best-known code for its parameter set \cite{table}. From the factorizations of $x^{33}-1$ and $x^{33}-2$ we observe that neither a cyclic nor a constacyclic code exists for length 33 and dimension 12. Moreover, since $33=11\cdot 3$, a 1-generator QT code with these parameters does not exist either. 

\end{Example}



\end{document}